\documentclass[11pt]{article}
\usepackage{comment}
\usepackage{amssymb}
\usepackage{subfigure}
\usepackage{amsmath}
\usepackage{amsfonts}
\usepackage{mathrsfs}
\usepackage{engord}
\usepackage[dvips]{graphicx}
\usepackage{bbm}
\usepackage{threeparttable}
\usepackage{booktabs}
\usepackage{epsfig}
\usepackage{color}
\usepackage{graphicx}
\usepackage{multirow}
\usepackage{cite}
\usepackage{epstopdf}
\usepackage{amsthm}
\usepackage{appendix}
\usepackage{framed}
\usepackage{url}

\begin{document}

\newtheorem{thm}{Theorem}
\newtheorem{prop}{Proposition}
\newtheorem{lem}{Lemma}
\newtheorem{defn}{Definition}
\newtheorem{ex}{Example}
\newtheorem{cor}{Corollary}
\newtheorem{prn}{Principle}
\newtheorem{case}{Case}
%
\title{Technical Report\\Classifying Unrooted Gaussian Trees under Privacy Constraints}

\author{Ali Moharrer,
Shuangqing Wei,
George T. Amariucai, 
and Jing Deng}
\maketitle
\footnotetext[1]{Part of this work is submitted to IEEE Globecom Conference, San Diego, CA, USA}
\footnotetext[2]{A. Moharrer, and S. Wei are with the school of Electrical Engineering and Computer Science, Louisiana State University, Baton Rouge, LA 70803, USA (Email: amohar2@lsu.edu, swei@lsu.edu). 

G. T. Amariucai is with the department of Electrical and Computer Engineering, Iowa State University, Ames, IA, USA (Email: gamari@iastate.edu). 

J. Deng is with the department of Computer Science, University of North Carolina at Greensboro, Greensboro, NC, USA (Email: jing.deng@uncg.edu).}



\begin{abstract}
In this work, our objective is to find out how topological and algebraic properties of unrooted Gaussian tree models determine their security robustness, which is measured by our proposed max-min information (MaMI) metric. Such metric quantifies the amount of common randomness extractable through public discussion between two legitimate nodes under an eavesdropper attack.  We show some general topological properties that the desired max-min solutions shall satisfy. Under such properties, we develop conditions under which comparable trees are put together to form partially ordered sets (posets).  Each poset contains the most favorable structure as the poset leader, and the least favorable structure. Then, we compute the Tutte-like polynomial for each tree in a poset in order to assign a polynomial to any tree in a poset.  Moreover, we propose a novel  method, based on restricted integer partitions, to effectively enumerate all poset leaders. The results not only help us understand the security strength of different Gaussian trees, which is critical when we evaluate the information leakage issues for various jointly Gaussian distributed measurements in networks, but also provide us both an algebraic and a topological perspective in grasping some fundamental properties of such models.  
\end{abstract}

\section{Introduction}

In this work, we are interested in the problem of effectively extracting maximum amount of common randomness through public discussions  between Alice and Bob. 
 based on their locally measured and correlated random variables \cite{maurer,part1,part2}.  Such a goal shall be attained 
in the presence of an eavesdropper, 
Eve, whose observations are also correlated with those possessed by Alice and Bob.
Such scenarios are pervasive in cases where a secret key is to be established between Alice and Bob by tapping into randomness 
available in their surrounding physical world. For example, in a sensor network with $n$ nodes, whose local readings on, for instance, temperature or humidity,  
are dependent following certain joint probability distribution 
function, Alice and Bob have to decide which two variables $A$ and $B$ 
out of $N$ nodes  are to be selected, whose realizations are to be used for building a secret key against 
a passive attack by Eve. The eavesdropper has full access to $Z$,  one of the remaining $N-2$ nodes/variables, as well as those publicly exchanged messages between Alice and Bob,
who establish secrecy by following the protocols proposed in \cite{maurer,part1,part2}
including both information reconciliation and privacy amplification stages. It has been already shown in \cite{maurer,part1,part2} for a given selected
three variables $A$, $B$ and $Z$, the number of bits per channel use that are information theoretically secure is proved to be $I(A;B|Z)$,  the conditional mutual information 
between $A$ and $B$, given $Z$.  

In this work,  we are particularly interested in the following two questions: (1) If the joint probability distribution of $n$ random variables is 
representable using certain graphical models, what variables Alice and Bob should pick, subject to Eve's selection of $Z$. (2) How should we compare and evaluate the strength 
of multiple graphical models in terms of the need of extracting secret key bits by Alice and Bob?
To address the questions raised above, we have adopted a pessimistic approach in that Alice and Bob move first by choosing the two variables $A$ and $B$ out of $n$ variables, and 
thereafter Eve chooses the variable $Z$ from the remaining ones to minimize the resulting conditional mutual information. Consequently, the selection of 
$A$ and $B$ is thus to maximize the minimum value, which yields the solution to the corresponding maxmin problem, thereby providing the answer to the first question, 
under a  given graphical model. For the second question, we compare different graphical models based on their respective maxmin values of the conditional mutual
information. 
It should be noted that such a modeling has been coined as the \textit{security game} in several contexts \cite{game}.
In fact, the authors in \cite{game} define the \textit{secrecy capacity} metric similar to
our defined metric, which quantifies the maximum rate of reliable information transmitted from the 
source to destination, in presence of the eavesdropper.

Due to the vast parameter space of graphical models, we restrict our attention in this work to a set of joint probability distribution functions of $n$ variables whose conditional
independence relationships  can be featured in Gaussian trees to address the aforementioned problems. 
The Gaussian models have been extensively used in a variety of topics. 
In fact, recently some fundamental properties of Gaussian graphical models
have been tackled using algebraic methods
\cite{sullivant2008}, \cite{roozbehani2014}. In \cite{sullivant2008}
the author shows that when the underlying random variables are Gaussian,
conditional independence statements can be interpreted as algebraic
constraints on the parameter space of the global model. 
Also, in our previous study
\cite{wcnc}, we proved that under the above assumptions, Alice, Bob and Eve
form special relations with respect to each other.

To address the question of comparing different Gaussian trees in terms of their associated maximin values of the conditional mutual information, we first
impose a constraint on the set of joint distributions we consider by requiring them to share the same joint entropy, i.e. the same total amount of randomness, and then
we need to extensively study some fundamental properties of certain classes
of unrooted Gaussian tree models related to our proposed security and privacy metrics.
In particular, we propose a \textit{grafting} operation \cite{graft}, to transform one Gaussian tree to
another by moving specific edges. Using grafting, we establish
a binary relationship between Gaussian trees to determine
the level of privacy, and further obtain an ordering for Gaussian
trees.

In \cite{wcnc} we showed that for Gaussian trees with $n=\{4,5\}$ nodes, the \textit{Linear}
model has the largest maximin value, hence it is the most secure model for smaller size networks, which enables us to attain  
a full ordering for both of the cases. 
In this paper, we consider any general class of Gaussian tree. 
We prove that unlike the cases of small size Gaussian
trees, for $n\geq 6$, not all the structures can be fully ordered using
grafting operation. Hence, we propose \textit{partially ordered sets}
(posets) \cite{poset} containing tree structures, where some of the
structures can be compared with each other using the binary relationship
and the others are not comparable. 
Moreover, in order to model the Gaussian trees and the corresponding posets
more systematically, we also study the algebraic properties of
unrooted Gaussian tree models. 
Lastly, For MF trees from all posets, we show that
enumerating these poset leaders can be related to integer partitions
\cite{integer}. 
In particular, under a set of proposed principles, we
can efficiently and directly enumerate all poset leaders without going
through the iterative grafting operations. 
These structures, are specifically
important, since they are the most secure trees in their related posets.

The paper is organized as follows. Section \ref{system model} presents
the system model. We study the topological and algebraic properties
of Gaussian trees in Sections \ref{topo} and \ref{algebra}. Section
\ref{conclusion} gives the concluding remarks.

\section{System Model} \label{system model}

In this study, we consider the Gaussian joint density to capture the density
of $n$ entities in a public channel, \textit{i.e.},
$P_{\boldsymbol{x}}(x_1,x_2,...,x_n)\sim \emph{N}(\mu,\Sigma)$,
where $\mu=0$ is the mean vector and $\Sigma$ is the symmetric,
positive-definite covariance matrix of $n$ random variables. Furthermore,
we assume that the joint density can be characterized by a weighted
and unrooted tree $T=(V,E,W)$, where $V$ is the set of vertices/variables,
and $E$ is the set of edges showing the dependency relations between variables
\cite{simecek2006,sullivant2008,roozbehani2014}.  For a fair
comparison between any two Gaussian tree, we assume that the users in
all models have the same joint randomness, i.e., the same entropy.
In this case, it is shown that the entropy of $\boldsymbol{x}=(x_1,x_2,...,x_n)$ can be obtained
by $H=1/2\ln((2\pi e)^n|\Sigma|)$ \cite{cover2012}. Hence, in order to obtain a fixed entropy we have to
fix the determinant of the covariance matrix, i.e., $|\Sigma|=k$, 
where $k\in \mathbb{R}$ is a finite and non-zero constant.
In addition, similar as in \cite{simecek2006} we assume normalized
diagonal entries for the covariance matrix.  
This assumption, simplifies the subsequent 
analysis, and as we will observe, gives rise to a more 
compact and manageable results.
On the other hand, note that normalization of
diagonal entries does not change the dependency relations in a Gaussian
tree. This condition makes the off-diagonal entries of $\Sigma$
to be in the range $(-1,1)$, and $k\in (0,1)$.  Also, given a  covariance matrix with
normalized diagonal entries, the edge-weights of the Gaussian tree $T$
correspond to the covariance elements in $\Sigma$. In particular,
for any edge $e_{ij}\in E$, and its associated weight $w_{ij}\in
W$, we have $w_{ij}=\sigma_{ij}$, where $\sigma_{ij}$ is the $(i,j)$-th element of the covariance matrix $\Sigma$.
Also, under these conditions, the
covariance value between any given pair of nodes $i,j\in V$, given
their connecting path $p_{ij}=e_{i,i+1}e_{i+1,i+2}...e_{j-1,j}$, is
$\sigma_{ij}=w_{i,i+1}w_{i+1,i+2}...w_{j-1,j}$. In other words, the
covariance value $\sigma_{ij}$ is the product of edge-weights on the
path $p_{ij}$ \cite{chaudhuri2014}.

Here, we first give a proper definition for the max-min problem, under the
explained scenario.

\begin{defn}
Under the Gaussian tree model, legitimate entities Alice and Bob pick two
nodes $a$ and $b$ on the tree under the attack by an eavesdropper Eve who
selects the third and distinct node/variable $z$ on the same tree. The
objective of Alice and Bob is to select the pair $(a,b)$ to maximize
the minimum conditional mutual information $I(a;b|z)$. As a result,
we adopt $\max_{\{a,b\}} \min_z I(a;b|z, T)$ as a metric to measure the
privacy level of a given weighted Gaussian tree $T=(V, E, W)$. 
\end{defn}

For Gaussian random variables the conditional mutual information
$I(a;b|z)$ can be directly related to the \textit{partial correlation
coefficient}, which is defined as below \cite{chaudhuri2014},
\begin{align} \label{eq:partial1}
\rho^2_{ab|z}=\dfrac{(\sigma_{ab}-\sigma_{az}\sigma_{bz})^2}{(\sigma_{aa}-\sigma_{az}\sigma_{az})(\sigma_{bb}-\sigma_{bz}\sigma_{bz})}=1-e^{-2I(a;b|z)}
\end{align}
\noindent where $\sigma_{ab}=E[(a-\mu_a)(b-\mu_b)]$, the $(a,b)$-th element of
$\Sigma$, is the covariance value between variables $a$ and $b$ (with both
of them having zero mean).  From \eqref{eq:partial1}, we can see that the
conditional mutual information is a monotone increasing function of the
partial correlation coefficient. Hence, in the following, we use partial
correlation coefficient instead of the conditional mutual information
as the security and privacy metric.  Hence, the corresponding max-min
value for a given Gaussian tree $T=(V,E,W)$  can be restated as:
\begin{align} \label{eq:max-min}
S(T,W)=\max_{\{a,b\}} \min_z \rho^2_{(ab|z,T,W)}
\end{align}

\noindent which is used to compare and order different trees.

\section{Topological Properties of Gaussian Trees} \label{topo}

We first define the \textit{grafting} operation on Gaussian trees.
In \cite{graft}, the author proposes an operation called \textit{grafting} to 
order trees based on their \textit{algebraic connectivity}, which is basically
the second smallest eigenvalue ($\lambda_2$) of the Laplacian matrix. 
Here, we introduce a new operation on Gaussian trees to obtain the ordering 
among different structures. Since, our proposed operation is similar to the 
\textit{grafting} operation introduced in \cite{graft} (but, obviously in a totally
different concept), we use the same naming to define our favorite operation.
\begin{defn} \label{defn:graft}
Consider a tree $T$, and assume there exists a leaf edge $e$,
between the vertices $n_1$ and the leaf $n_2$. The node $n_1$ has degree two.
The grafting operation refers to cutting
the edge $e$ and attaching it to the other end of its parent edge,
i.e. $e'$, as shown in Figure \ref{fig:grafting}.  
\end{defn}

Note that since grafting is essentially a local operation, only the edge
$e$ changes its position: In $T_1$ there is an edge $e$ between the pair
$\{n_1,n_2\}$, while in $T_2$ this edge is between the pair $\{v,n_2\}$. All
other structures shown in Figure \ref{fig:grafting} (including everything
in the clouds) remain unchanged.

\begin{figure} [h!]
\centering 
\includegraphics[scale=0.5]{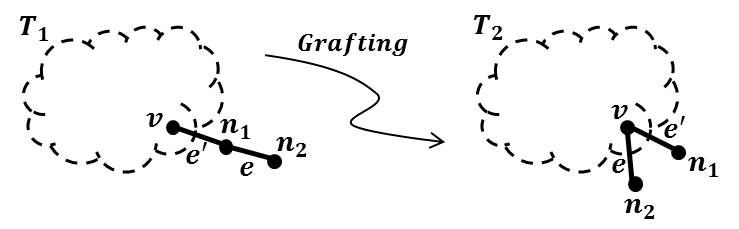}
\caption{$T_2$ is obtained from $T_1$ by grafting operation\label{fig:grafting}} 
\end{figure}

In \cite{wcnc} we showed that for any Gaussian tree the maximin value of the conditional mutual information $I(a;b|z)$ is chosen from those set of triplets in which $a$ and $b$ are adjacent and $z$ is neighbor to either $a$ or $b$. 
This result may seem intuitive; however, we showed that for directed Gaussian trees, there are many cases in which this
result does not hold. 

The computed partial correlation coefficient has the following form,
\begin{align} \label{eq:partial2}
\rho^2_{ab|z}=\dfrac{\sigma^2_{ab}(1-\sigma^2_{bz})}{1-\sigma^2_{ab}\sigma^2_{bz}}
\end{align}

\noindent where the covariance values $\sigma_{ab}$ and $\sigma_{bz}$
are the corresponding weights for the edges $e_{ab}$ and $e_{bz}$,
respectively. Note that in \eqref{eq:partial2} we implicitly assumed
$z$ is adjacent to $b$, hence if $z$ becomes adjacent to $a$, then
$\sigma_{bz}$ replaces with $\sigma_{az}$ in the equation.
We also proved the following result in \cite{wcnc},

\begin{lem} \label{lem:lemwcnc}
Consider the trees $T_1$ and $T_2$ shown in Figure \ref{fig:grafting},
both trees have the same set of labeling and edge weights except that
the label for the weight of the edge $e$ is switched from $w_{n_1,n_2}$ to $w_{v,n_2}$. More
precisely, for $T_2=(V,E',W')$ that is obtained from $T_1=(V,E,W)$
using grafting operation, we know all the elements in $E'$ and
$W'$ are the same as in $E$ and $W$, respectively, except the entry
corresponding to the edge $e$. For this element we have $e'_{v,n_2}=e_{n_1,n_2}$
and $w'_{v,n_2}=w_{n_1,n_2}$. Now, suppose the maximin value for $T_1$, and
$T_2$ are $S(T_1,W)$ and $S(T_2,W')$, respectively. Note that $W$
is any arbitrary set of edge-weights, and $W'$ is obtained from $W$
(by changing the covariance values associated with the grafted edge). We
have $S(T_1,W)\geq S(T_2,W')$.
\end{lem}

As we can see from Lemma \ref{lem:lemwcnc}, for any given tree structure
with edge-weights chosen from the corresponding entries of the covariance
matrix $\Sigma$, the grafting operation always decreases the maximin
value of the resulting tree.  In fact, by grafting the edge $e$ we are
essentially adding another neighbor to the node $n_2$. This in turn,
gives more options to eavesdropper to choose the best location to attack,
resulting in smaller maximin values.  As a result, grafting operation
generates a certain ordering of trees, in which the corresponding
topologies are ordered with regard to their respective maximin values. In
the following, we formally define the tree ordering using the results
obtained in Lemma \ref{lem:lemwcnc},

\begin{defn} \label{defn:order}
Consider the trees $T_1=(V,E,W)$ and $T_2=(V,E',W')$, where $T_2$
is obtained from $T_1$ using grafting operation.  We know from Lemma
\ref{lem:lemwcnc} that $S(T_1,W)\geq S(T_2,W')$. In this setting, we
write $T_1\succeq T_2$, where the binary relation $``\succeq"$ shows
the ordering of these trees, i.e., $T_1$ is more favorable than $T_2$.
\end{defn}

As we will see shortly, the ordering, which is defined in 
Definition \ref{defn:order} gives rise to developing
an interesting concept: we define several \textit{classes} for all Gaussian trees with
the same order. Essentially, each class is a particular poset of distinct
Gaussian trees.

\subsection{The general cut and paste operation}

We proved that for Gaussian trees with $n\in\{4,5\}$ nodes the linear topology always has largest maximin value, so it is the most secure structure \cite{wcnc}. However, in this paper we prove that for $n\geq 6$ this is not the case.
We have the following results,

We have the following result,

\begin{lem} \label{lem:covariance}
\emph{
Consider any Gaussian tree $T=(V,E,W)$, with order $|V|=n$. We denote $|\Sigma_T|$ as the determinant of covariance matrix corresponding to $T$.
Considering the PG operation shown in Figure \ref{fig:generalcut}, which transforms the Gaussian tree $T_1$ into $T_2$, with $|\Sigma_{T_1}|=|\Sigma_{T_2}|$. Let us denote $\sigma_{n_1n_2}$ and $\sigma'_{v'n_2}$ as the covariance values between the pairs $(n_1,n_2)$ and $(v',n_2)$ in trees $T_1$ and $T_2$, respectively; then we have $\sigma^2_{n_1n_2}/\sigma_{n_1n_1} = \sigma'^2_{v'n_2}/\sigma_{vv}$.
}
\end{lem}

\begin{proof}
See Appendix \ref{app:covariance}.
\end{proof}

\begin{prop} \label{prop:generalcut}
Consider the trees shown in Figure \ref{fig:generalcut}. Given a Gaussian
tree $T_1=(V,E,W)$, with leaf edge $e$, which is connected to $e'$ if
we cut $e$ and paste it to some vertex other than $v$ (unlike grafting),
say $v'$, we obtain the Gaussian tree $T_2=(V,E',W')$. Then, in general
$T_1\nsucceq T_2$. Hence, in general $T_1$ is not always more favorable
than $T_2$.  
\end{prop}

\begin{proof}
See Appendix \ref{app:generalcut}.
\end{proof}

\begin{figure} [h!]
\centering
\includegraphics[scale=0.5]{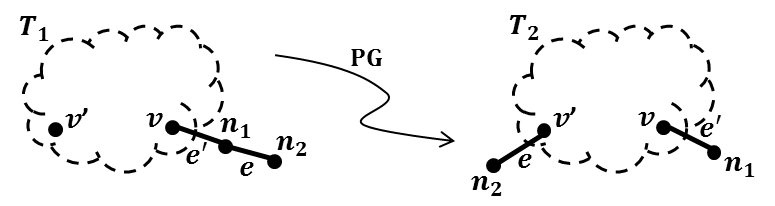} 
\caption{The figure showing the conditions in Proposition \ref{prop:generalcut}}\label{fig:generalcut} 
\end{figure}

From Proposition \ref{prop:generalcut} we can see that if two trees are not related through one or more grafting operations, then in general they cannot be ordered using our defined binary relation. In fact, without assigning a specific covariance matrix (hence the set of edge-weights) these structures cannot be consistently compared. Since not all tree topologies can be obtained from the linear structure using grafting, so the linear topology is not always the most secure structure. This result motivates us to seek for certain topologies that cannot be compared with each other, and at the same time they cannot be improved further, using grafting operation. In particular, we form sets of tree structures, where each set contains a unique leader that is the most favorable topology among all other topologies in the same set. Other topologies in a poset might be comparable/incomparable with each other. By classifying the trees into certain sets we can further study both their topological and algebraic properties.

\subsection{Forming the posets of Gaussian trees}

Based on the obtained results in Proposition \ref{prop:generalcut} we
can form posets \cite{poset} of Gaussian trees. Each poset is formed
from its most favorable (MF) structure, $T_M=(V_M,E_M,W_M)$. In other
words, $T_M$ is the poset leader acting as the \textit{ancestor} to
all other Gaussian trees in a poset, i.e., all other Gaussian trees can
be obtained from $T_M$ using one or more grafting operations. Also, in
each poset given two trees $T_i$ and $T_j$, they are adjacent if $T_j$
can be obtained from $T_i$ by grafting one of its edges. In this case,
there is a directed edge from $T_i$ to $T_j$, i.e. $T_i \rightarrow T_j$.

\begin{lem} \label{lemma:LF}
In any poset with a given $T_M=(V_M,E_M,W_M)$ acting as a poset leader,
we can find a unique least favorable (LF) structure, $T_L=(V_L,E_L,W_L)$,
which acts as a \textit{descendant} to all other trees.
\end{lem}

\begin{proof}
See Appendix \ref{app:LF}
\end{proof}

Hence, we observe that our defined posets are certain class of posets,
which have a unique MF ans LF structures. Also, from the results in
Lemma \ref{lem:lemwcnc} we know that $T_M$ has the most secure structure,
while $T_L$ has the least secure structure in each poset.
As an example, Figure \ref{fig:posets} shows all three posets of Gaussian
trees on $7$ nodes.  It can be seen that when there is a directed path
between two topologies, then they are comparable using our defined
binary relation.  Note that in this figure, posets $1$ and $3$ are
the special cases where posets are basically formed as fully ordered
sets, hence any tree structure in each of these poset can be compared
to other trees in the same poset.  On the other hand, in poset $2$,
the two middle structures cannot be compared using the rules given in
Lemma \ref{lem:lemwcnc}.  Also, the MF and LF structures are placed at
the top and bottom of each poset, respectively. 

\begin{figure}[h!]
\centering
\includegraphics[scale=0.40]{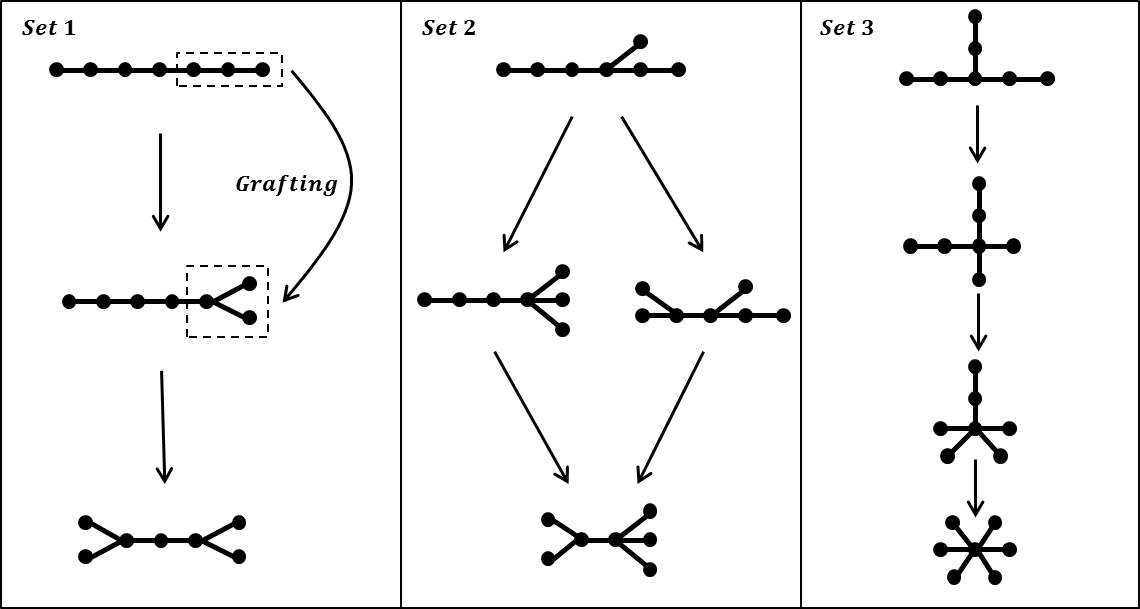} 
\caption{All the possible posets for Gaussian trees with $n=7$ nodes\label{fig:posets}} 
\end{figure}

The MF topologies are the most secure trees in each poset; also as we
observed, poset leaders can characterize all other structures in a poset:
the poset leaders can fully describe the poset structure.
Hence, finding such structures is of
huge importance. However, there should be a method to systematically
obtain these topologies. Thus, in section \ref{algebra} we propose an
efficient way, which can enumerate all these structures systematically.

\subsection{Forming the super-graph for each poset}
Figure \ref{fig:posets} gives us an intuition in order to construct
a directed super-graph containing Gaussian trees. In particular, each
poset can be converted into a directed super-graph $G=(V_s,E_s)$, where
$V_s$ is the set of trees in a poset acting as vertices, and $E_s$ is
the set of directed edges between the two nodes that can be related
using grafting.  Using this super-graph, we can easily identify the
comparable tree structures: If there is a directed path between two
structures, then they are comparable. Hence, we can conclude that both
MF and LF structures can be compared with any other tree in a poset.
Also, observe that poset leader fully characterizes the structure of
its super-graph. In particular, the number of those leaf edges (in 
the poset leader structure) that are adjacent to a particular node with degree 
two, specifies the \textit{length} (number of
consecutive grafting operations plus $1$) of the super-graph. Moreover,
the structure of those special edges specifies the \textit{width} of the
super-graph. In particular, consider the following:
in Figure \ref{fig:posets} we can see that the
poset $3$ have three special edges, hence the super-graph has length
$4$. Also, since these special edges are fully \textit{symmetric} with respect
to each other (grafting either of those edges, results in a same tree),
so the poset $3$ becomes fully ordered. On the other
hand, in poset $2$ because of the two \textit{asymmetric} branches we obtain two
different topologies in the next level. In general, if those special
branches become more symmetric, the poset tends to become fully ordered.

Although, converting each poset to its corresponding super-graph
simplifies the comparison between tree topologies in a set, but as it
can be observed, for larger trees identifying these special branches and
ordering trees by grafting operation becomes more challenging. Hence, in
the following we aim to study the tree structures and their associated
posets in a more abstract and general ways.

\section{Algebraic properties of Gaussian Trees} \label{algebra}

\subsection{Tutte-Like Polynomials for Gaussian Trees in Posets}

In this section, in order to model the  Gaussian trees and the
corresponding posets more systematically, we study the algebraic
properties of these models. As we may see in the following, these
properties will further help us characterize those special leaf nodes with
no leaf siblings, and thus allow us to evaluate the security robustness
of any tree structure within a poset. To achieve this goal, for each
tree,  we associate a two-variable \textit{Tutte-like} polynomial
defined in \cite{tutte}, where the authors modify the definition of the
Tutte polynomial to obtain a new invariant for both rooted and unrooted
trees. Also, they proved that this polynomial uniquely determines rooted
trees. For unrooted trees however, it is shown in \cite{nonisomo} that
certain classes of caterpillars have the same polynomials assigned to
them. However, interestingly, we prove that in each poset, 
in many cases the trees have unique polynomials.

Let $R(T)$ denote the collection of all subtrees of $T$, and $L_E(S)$
denote the leaf edges in the subtree $S$, i.e., the edges that are
connected to leaf nodes then \cite{tutte},
\begin{align} \label{eq:tutte}
f_E(T;t,z)=\sum_{S\subseteq R(T)} t^{|E(S)|}(z+1)^{|E(S)|-|L_E(S)|}
\end{align}

\noindent where $|E(S)|$ is the total number of edges in the
subtree $S$. Basically, this polynomial is a generating function
that encodes the number of subtrees with a given internal and leaf
edges \cite{nonisomo}. We next show that such polynomials can help us
systematically generate trees in a poset from the poset leader. The
proof can be found in Appendix \ref{app:recursive}.
\begin{lem} \label{lem:recursive}
Suppose there is a directed path from the tree $T_n$ to $T_{n-m}$ in a poset, i.e., $T_{n-m}$ can be obtained from $T_{n}$ through $m$ grafting operation.
Then, their associated polynomials have the following recursive formula,
\begin{align} \label{eq:recursive}
\notag f(T_n;&t,z) =\\ &f(T_{n-m};t,z)+ t(1-tz)[m-\sum_{k=1}^m g_{n-k}(t,z)]
\end{align}
where, $g_{n-k}(t,z)$ is the polynomial associated to the rooted tree obtained from the tree $T_{n-k}$, after deleting the edges $e$ and its neighbor edge $e'$ (e.g., see $e$ and $e'$ shown in Figure \ref{fig:grafting} for the tree $T_1$), in a given step $k$, and putting their common node as a root (e.g., the node $v$ in Figure \ref{fig:grafting}). Note that in \eqref{eq:recursive}, $T_1$ is the LF topology.
\end{lem}

Using the recursive equation derived in \eqref{eq:recursive}, we then have
the following corollary, whose proof is in Appendix \ref{app:unique}.
\begin{cor} \label{cor:unique}
In a poset, certain tree structures are uniquely determined by Tutte-like
polynomial. In particular, if one of the following cases happen then two
polynomials are distinct: (1) If there exists a directed path between
two trees; (2) If both trees have the same parent tree; (3) If the two
structures lie at different levels.
\end{cor}

Hence, by Corollary \ref{cor:unique} we see that although Tutte-like
polynomial is not graph invariant in general, but in many cases the
polynomials associated to trees in a same poset are distinct.  As an
example, consider the poset $2$ shown in Figure \ref{fig:posets}. Since
all trees satisfy at least one of the conditions in Corollary
\ref{cor:unique}, all of their associated polynomials are thus 
be distinct. Following  
\eqref{eq:tutte}, we have 
\begin{align} \label{eq:tuttepolyposet2}
f(T_M;t,y)=t^6y^3&+t^5(y^3+2y^2)+t^4(3y^2+2y)\notag\\&+t^3(5y+1)+6t^2+6t+1\notag\\
f(T_l;t,y)=t^6y^2&+t^5(3y^2+y)+t^4(3y^2+3y+1)\notag\\&+t^3(4y+4)+8t^2+6t+1\notag\\
f(T_r;t,y)=t^6y^2&+t^5(3y^2+y)+t^4(2y^2+5y)\notag\\&+t^3(6y+2)+7t^2+6t+1\notag\\
f(T_L;t,y)=t^6y&+5t^5y+t^4(8y+1)+t^3(6y+5)\notag\\&+9t^2+6t+1
\end{align} 

\noindent where $T_M$ and $T_L$ are the MF and LF topologies in poset $2$,
respectively. Also, $T_l$ and $T_r$ are the left and right topologies,
respectively that located in the middle of poset $2$. For the simplicity
of polynomials we replaced $z+1$ with $y$. As we expect, all the computed
polynomials in \eqref{eq:tuttepolyposet2} are distinct.

The Tutte-like polynomial can be used to evaluate certain topological
properties of trees.   In the following lemma, whose proof is in
Appendix \ref{app:specialbranches}, we propose an interesting result: 
the Tutte-like polynomial can enable us to extract the exact number 
of those special leaf edges from this polynomial. Hence, using 
this result we estimate the security robustness of a tree structure 
by computing its distance from LF structure.

\begin{lem} \label{lem:specialbranches}
Given the polynomial $f(T;t,z)$ associated with $T$, the second highest degree term has the form $t^{|E|-1}(\alpha(1+z)^{|I|-1}+\beta(1+z)^{|I|})$. The coefficient $\alpha$ shows the number of leaf edges with no leaf siblings.
\end{lem}

\begin{cor} \label{cor:distance}
The coefficient $\alpha$ defined in Lemma \ref{lem:specialbranches} shows the distance between the tree $T$ and LF structure. Also, if $\alpha=0$ then $T$ is the LF structure. 
\end{cor}

\begin{ex}
Consider the tree topologies in poset $2$ of Figure \ref{fig:posets}, and their associated polynomials that is computed in \eqref{eq:tuttepolyposet2}. The MF tree $T_M$ has two leaf edges with no siblings, hence in its corresponding polynomial, the second highest degree term has the form $t^5y^3+2t^5y^2$. Hence, $\alpha=2$. On the other hand, the LF tree $T_L$ has no such leaf edges. From \eqref{eq:tuttepolyposet2} we can see the second highest degree term for $f(T_L;t,y)$ is $5t^5y$, hence $\alpha=0$.
\end{ex}

The results obtained in Lemma \ref{lem:specialbranches} and Corollary \ref{cor:distance} show the strong correlation between the Tutte-like polynomial and security robustness of Gaussian tree. This information is very helpful in order to compare the security of a tree in a poset. In particular, being closer to LF structure, hence having smaller values for $\alpha$ (comparing to others in the same poset), makes the Gaussian tree less favorable comparing to other structures in a poset.

\subsection{Enumerating Poset Leaders: Restricted Integer Partition Approach}

In the previous sections, we studied certain properties of tree
topologies in the same poset.  In this section, we find a systematic
way to generate different poset leaders, which is further related to
\textit{restricted} integer partition problems.  The method we propose
is essentially determined by the property of each MF in that the leaf
edges have no leaf siblings in all MF structures.  The following example
will demonstrate new ways to quickly enumerate these MF models.

Consider the MF topology in poset $3$ shown in Figure
\ref{fig:posets}. There are three branches coming out of the central
node, which has degree $3$. Each branch has $2$ nodes, hence we assign
the string $(2+2+2)$ to this topology. Each number (here all numbers
are $2$) shows the length of their corresponding branches. Note that
we do not count the central node, which we name it the \textit{anchor}
node. Similarly consider the MF topology in poset $2$, which is shown in
Figure \ref{fig:posets}. Here, the anchor node is the vertex with degree
$3$, hence we assign the string $(1+2+3)$ to this topology. Lastly,
consider the MF topology in poset $1$. Since, all the internal nodes,
can be anchor nodes, hence we can assign multiple (equivalent) strings to
this topology, i.e., $(1+5)$, $(2+4)$, and $(3+3)$ are all valid strings.

Based on this example, we propose an effective algorithm (in fact, the
set of constraints) to enumerate all poset leaders of given order. It
turns out that integer partition methods \cite{integer} can be very
helpful in order to quickly reach this goal. However, this method should
be systematically implemented. In particular, we use \textit{restricted}
integer partitions to find all poset leaders. 

Each integer partition should satisfy the following conditions: (1) Each
part should have at most a single $1$ (2) The leftmost and rightmost parts
should both be larger than $3$.  Given $|V|=n$ we do the following until
all of the parts cannot be partitioned without violating the restrictive
conditions:

\begin{framed}
\textbf{Enumerating Poset Leaders}

\-\hspace{0.2cm}$1.$ Find all integer partitions for $n-A$

\-\hspace{0.2cm}$2.$ For each one of these partitions, find those 
parts that can be further partitioned, and follow the 
steps in $1$.

\-\hspace{0.2cm}$3.$ Check for any redundant partitions and eliminate  them, and if any permutation of parts gives a new 
poset leader structure
\end{framed}

Here $A$ shows the number of anchor nodes. Basically, in the above set of principles the first constraint is to ignore the non-poset leader cases, while the second constraint is to ignore the cases where two or more parts can be merged and form already produced parts, hence, making this method more effective. 
The anchor nodes are certain non-leaf vertices, acting as a \textit{hub} for two or more branches. Each of the anchor nodes, with their associated branches can form a smaller integer partition satisfying the aforementioned constraints.
Also, unlike normal integer partitions the position of parts matters, so we should count some of permutations of different parts. In particular, two non-isomorphic poset leader topologies may have identical integer partitions, but with different ordering of parts.



\section{Conclusion} \label{conclusion}
In this paper, we analyzed the information leakage in public communications channels, where the joint density of entities in the channel can be modeled by Gaussian trees. We addressed two fundamental problems: $(1)$ In the max-min scenario we showed the special relation between $3$ variables that are be chosen by Alice, Bob, and Eve. $(2)$ Under the same scenario, we studied the impact of choosing different structures, on the maximin value and hence on the channel security.  
We proposed the grafting operation, which produces less favorable trees. Then, we ordered the tree structures using our defined pair-wise relationship. Interestingly, using our defined operation, we obtained partially ordered sets of trees, through which we classified all the Gaussian tree structures of given order into several partially ordered sets. We proved a particular feature for the sets: each poset have a unique MF and LF structures. In order to further simplify comparing  topologies, we modeled each poset as a directed super-graph, by which  any two trees are comparable if there is a directed path from one to the other.  
Moreover, we provided a systematic way of producing all trees in a poset through computing corresponding Tutte-like polynomials. We obtained certain fundamental results:  using the second highest degree term in each polynomial one can evaluate the security robustness of the given structure. Lastly, we introduced a restricted integer partition approach based on proposed principles to quickly enumerate all poset leaders of a given order without listing all non-leader structures.

\bibliography{reference}

\begin{thebibliography}{10}
\providecommand{\url}[1]{#1}
\csname url@samestyle\endcsname
\providecommand{\newblock}{\relax}
\providecommand{\bibinfo}[2]{#2}
\providecommand{\BIBentrySTDinterwordspacing}{\spaceskip=0pt\relax}
\providecommand{\BIBentryALTinterwordstretchfactor}{4}
\providecommand{\BIBentryALTinterwordspacing}{\spaceskip=\fontdimen2\font plus
\BIBentryALTinterwordstretchfactor\fontdimen3\font minus
  \fontdimen4\font\relax}
\providecommand{\BIBforeignlanguage}[2]{{%
\expandafter\ifx\csname l@#1\endcsname\relax
\typeout{** WARNING: IEEEtran.bst: No hyphenation pattern has been}%
\typeout{** loaded for the language `#1'. Using the pattern for}%
\typeout{** the default language instead.}%
\else
\language=\csname l@#1\endcsname
\fi
#2}}
\providecommand{\BIBdecl}{\relax}
\BIBdecl

\bibitem{maurer}
U.~M. Maurer, ``Secret key agreement by public discussion from common
  information,'' \emph{Information Theory, IEEE Transactions on}, vol.~39,
  no.~3, pp. 733--742, 1993.

\bibitem{part1}
R.~Ahlswede and I.~Csisz{\'a}r, ``Common randomness in information theory and
  cryptography. part i: secret sharing,'' \emph{IEEE Transactions on
  Information Theory}, vol.~39, no.~4, 1993.

\bibitem{part2}
------, ``Common randomness in information theory and cryptography. ii. cr
  capacity,'' \emph{Information Theory, IEEE Transactions on}, vol.~44, no.~1,
  pp. 225--240, 1998.

\bibitem{game}
M.~H. Manshaei, Q.~Zhu, T.~Alpcan, T.~Bac{\c{s}}ar, and J.-P. Hubaux, ``Game
  theory meets network security and privacy,'' \emph{ACM Computing Surveys
  (CSUR)}, vol.~45, no.~3, p.~25, 2013.

\bibitem{sullivant2008}
S.~Sullivant, ``Algebraic geometry of {Gaussian Bayesian} networks,''
  \emph{Advances in Applied Mathematics}, vol.~40, no.~4, pp. 482--513, 2008.

\bibitem{roozbehani2014}
H.~Roozbehani and Y.~Polyanskiy, ``Algebraic methods of classifying directed
  graphical models,'' \emph{arXiv preprint arXiv:1401.5551}, 2014.

\bibitem{wcnc}
A.~Moharrer, S.~Wei, G.~Amariucai, and J.~Deng, ``Evaluation of security
  robustness against information leakage in {Gaussian} polytree graphical
  models,'' in \emph{Proceedings of the Wireless Communications and Networking
  Conference (WCNC), 2015 IEEE}, pp. 1422--1427.

\bibitem{graft}
K.~Patra and A.~Lal, ``The effect on the algebraic connectivity of a tree by
  grafting or collapsing of edges,'' \emph{Linear Algebra and its
  Applications}, vol. 428, no.~4, pp. 855--864, 2008.

\bibitem{poset}
W.~T. Trotter, \emph{Combinatorics and partially ordered sets: Dimension
  theory}.\hskip 1em plus 0.5em minus 0.4em\relax JHU Press, 2001, vol.~6.

\bibitem{integer}
G.~E. Andrews, \emph{The theory of partitions}.\hskip 1em plus 0.5em minus
  0.4em\relax Cambridge university press, 1998, vol.~2.

\bibitem{simecek2006}
P.~{\v{S}}imecek, ``Gaussian representation of independence models over four
  random variables,'' in \emph{COMPSTAT conference}, 2006.

\bibitem{cover2012}
T.~M. Cover and J.~A. Thomas, \emph{Elements of information theory}.\hskip 1em
  plus 0.5em minus 0.4em\relax John Wiley \& Sons, 2012.

\bibitem{chaudhuri2014}
S.~Chaudhuri, ``Qualitative inequalities for squared partial correlations of a
  {Gaussian} random vector,'' \emph{Annals of the Institute of Statistical
  Mathematics}, vol.~66, no.~2, pp. 345--367, 2014.

\bibitem{tutte}
S.~Chaudhary and G.~Gordon, ``Tutte polynomials for trees,'' \emph{Journal of
  graph theory}, vol.~15, no.~3, pp. 317--331, 1991.

\bibitem{nonisomo}
D.~Eisenstat and G.~Gordon, ``Non-isomorphic caterpillars with identical
  subtree data,'' \emph{Discrete mathematics}, vol. 306, no.~8, pp. 827--830,
  2006.

\end{thebibliography}
\bibliographystyle{IEEEtran}


\newpage
\appendices

\section{Proof of Lemma \ref{lem:covariance}} \label{app:covariance}

We first prove the following:

For any $v_i\in V$, let us define $d_i$ as its degree. Then, we have,
\begin{align} \label{eq:cov_general}
		|\Sigma_T| = \dfrac{\prod\limits_{(v_i,v_j)\in E}[\sigma_{v_iv_i}\sigma_{v_jv_j}-\sigma^2_{v_iv_j}]}{\prod\limits_{v_i\in V}\sigma_{v_iv_i}^{d_i-1}}
\end{align}
The proof follows by induction. First, assume that the Gaussian tree $T$ has only one edge with two vertices $v_1$ and $v_2$, then we can immediately form $\Sigma_T$, and deduce 
$|\Sigma_T| = \sigma_{v_1v_1}\sigma_{v_2v_2}-\sigma^2_{v_1v_2}$, which follows the general formula in \eqref{eq:cov_general}. Next, let us assume that \eqref{eq:cov_general} is valid up to $T'=(V',E',W')$, where $|V'|=n-1$. Hence, we need to prove the validity of this equation for $T=(V,E,W)$ with $|V|=n$, where the tree $T$ can be obtained from $T'$ by adding one vertex, namely $v_n$. 
Without loss of generality, we assume that $v_n$ is connected to $v_{n-1}$. 
Since $v_n\perp v_i|v_{n-1}$, for all $v_i\in V\backslash\{v_n,v_{n-1}\}$, using the arguments given in Section \ref{system model}, we have $\sigma_{v_iv_n}=\sigma_{v_iv_{n-1}}\sigma_{v_{n-1}v_n}/\sigma_{v_{n-1}v_{n-1}}$. If we factorize ${\sigma_{v_{n-1}v_n}}/{\sigma_{v_{n-1}v_{n-1}}}$ from the last column, then subtract the $n-1$-th column from the $n$-th column, and replace the result with the $n$-th column, we obtain,

\begin{align*}
|\Sigma_T| = \dfrac{\sigma_{v_{n-1}v_n}}{\sigma_{v_{n-1}v_{n-1}}}
\begin{vmatrix}
\sigma_{v_1v_1} & \cdots & 0\\
\sigma_{v_2v_1} & \cdots & 0\\
\vdots & \ddots & \vdots\\
\sigma_{v_{n-1}v_1} & \cdots & 0\\
\sigma_{v_nv_1} & \cdots & x
\end{vmatrix}
\end{align*}

\noindent where $x=(\sigma_{v_nv_n}\sigma_{v_{n-1}v_{n-1}}-\sigma^2_{v_nv_{n-1}})/\sigma_{v_nv_{n-1}}$.
Using the last column, we can compute $|\Sigma_T|$ as follows,
\begin{align} \label{eq:sigmat}
|\Sigma_T| &= \dfrac{\sigma_{v_{n-1}v_n}}{\sigma_{v_{n-1}v_{n-1}}}x|\Sigma_{\backslash n\backslash n}|\notag\\
&= \dfrac{\sigma_{v_nv_n}\sigma_{v_{n-1}v_{n-1}}-\sigma^2_{v_nv_{n-1}}}{\sigma_{v_{n-1}v_{n-1}}}|\Sigma_{\backslash n\backslash n}|
\end{align}
\noindent where $|\Sigma_{\backslash n\backslash n}|$ is the determinant of submatrix of $\Sigma_T$ resulting after removing the $n$-th column and row. Note that since removing the last row and column of $\Sigma_T$ is the same as removing $v_n$ from $T$, hence we conclude that $\Sigma_{\backslash n\backslash n}=\Sigma_{T'}$. Therefore, \eqref{eq:sigmat} becomes
\begin{align} \label{eq:sigmat2}
|\Sigma_T| = \dfrac{\sigma_{v_nv_n}\sigma_{v_{n-1}v_{n-1}}-\sigma^2_{v_nv_{n-1}}}{\sigma_{v_{n-1}v_{n-1}}}|\Sigma_{T'}|
\end{align}
Observe that since the degree of $v_{n-1}$ is $2$, the fraction in \eqref{eq:sigmat2} has the same form as in \eqref{eq:cov_general}. Also, we know that $|\Sigma_{T'}|$ follows the general formula as well. This completes the inductive proof for the first part.

To prove the result in Lemma \ref{lem:covariance}, note that the Gaussian trees $T_1$ and $T_2$ differ in only one edge, namely $e$. Hence, we can write,
\begin{align*}
&|\Sigma_{T_1}| = |\Sigma_{\backslash n_2\backslash n_2}|\dfrac{\sigma_{n_1n_1}\sigma_{n_2n_2}-\sigma^2_{n_1n_2}}{\sigma_{n_1n_1}}\\
&|\Sigma_{T_2}| = |\Sigma_{\backslash n_2\backslash n_2}|\dfrac{\sigma_{v'v'}\sigma_{n_2n_2}-\sigma^2_{v'n_2}}{\sigma_{v'v'}}
\end{align*}
Since we assume that $|\Sigma_{T_1}|=|\Sigma_{T_2}|$, the result follows.

\section{Proof of Proposition \ref{prop:generalcut}} \label{app:generalcut}

Based on the values of edge-weights, two possible cases may happen for the maximin value of the tree $T_1$ and $T_2$ that are shown in Figure \ref{fig:generalcut}:

\begin{case} \label{case:1}
Let's assume $S(T_1,W)$ consists the nodes $n_1$, $n_2$, and $v$. In particular, using \eqref{eq:partial2} we obtain $S(T_1,W)=\dfrac{\sigma^2_{vn_1}(1-\sigma^2_{n_1n_2})}{1-\sigma^2_{n_1n_2}\sigma^2_{vn_1}}$. In other words, $S(T_1,W)$ is larger than all other partial correlations except the triplets which contain the pair $\{n_1,v\}$ and the nodes other than $n_2$. Now, the maximin value of $T_2$, say $S(T_2,W)$, is obtained by $S(T_2,W')=\max\{\min_i(\dfrac{\sigma^2_{vn_1}(1-\sigma^2_{vx_i})}{1-\sigma^2_{vx_i}\sigma^2_{vn_1}}),f(T_2,W')\}$, where $x_i$ are the adjacent nodes to $v$, other than $n_1$. Also, $f(T_2,W')$ can be any other value chosen from the max-min table corresponding to the tree $T_2$. We can show that, in any case we have $S(T_2,W')\geq S(T_1,W)$.
\end{case}
\begin{case} \label{case:2}
If $S(T_1,W)=\dfrac{\sigma^2_{v'x'_i}(1-\sigma^2_{x'_iy'_j})}{1-\sigma^2_{v'x'_i}\sigma^2_{x'_iy'_j}}$, where $x'_i$ is adjacent to $v'$ and $y'_j$ is adjacent to $x'_i$, then for $S(T_2,W')$, we have $S(T_2,W')=\min\{S(T_1,W),\dfrac{\sigma^2_{v'x'_i}(1-\sigma^2_{v'n_2})}{1-\sigma^2_{v'x'_i}\sigma^2_{v'n_2}}\}$. Hence, we can conclude that $S(T_1,W)\geq S(T_2,W')$.
\end{case}

From two above cases we can conclude that in Case \ref{case:1}, $T_2$ is more favorable than $T_1$, whereas in Case \ref{case:2} the opposite conclusion holds.

\section{Proof of Lemma~\ref{lemma:LF}} \label{app:LF}
The proof is quite straightforward. For any poset, given its poset
leader $T_M$ we can find and graft those leaf edges with no leaf
siblings. Obviously, we can begin with any list of edges with any
order. Hence, in any case we end up with a unique structure, namely $T_L$.

\section{Proof of Lemma \ref{lem:recursive}} \label{app:recursive}
By combining propositions $4$, $5$, $15$, and $18$ from \cite{tutte} we can conclude the following,
\begin{align*} 
&f(T_n;t,z) = h_{n-1}+ t[t(z+1)g_{n-1}+(1-tz)]\\
&f(T_{n-1};t,z) = h_{n-1}+ t(t+1)g_{n-1}
\end{align*}
where $h_{n-1}$ is the polynomial related to the unrooted tree $T_n/e$, i.e., the unrooted tree $T_n$ after deleting an edge $e$. Note that for simplicity we write $h_{n-1}$ and $g_{n-1}$ instead of $h_{n-1}(t,z)$ and $g_{n-1}(t,z)$, respectively. Using the above equations, we obtain the following recursive formula,
\begin{align}
f(T_n;&t,z) = f(T_{n-1};t,z)+ t(1-tz)[1-g_{n-1}(t,z)]
\end{align}
If we proceed up to level $m$, the result in \eqref{eq:recursive} follows.

\section{Proof of Corollary \ref{cor:unique}} \label{app:unique}
First, suppose $f(T_n;t,z) = f(T_{n-m};t,z)$ then using \eqref{eq:recursive} we should have $t(1-tz)[m-\sum_{k=1}^m g_{n-k}(t,z)]=0$, or $\sum_{k=1}^m g_{n-k}(t,z)=m$. Recall that all $g_{n-k}(t,z)$ are polynomials associated with rooted trees, so the only possibility is $g_{n-k}(t,z)=1$, for all $1\leq k\leq m$, a contradiction.

Second, consider two trees $T_{n-m}$ and $T_{n-l}$, at different levels having nearest common ancestor $T_n$. Then using \eqref{eq:recursive} we have the following:
\begin{align*}
f(T_n;&t,z) =\begin{cases} f^L(T_{n-m};t,z)+ t(1-tz)[m-\sum_{k=1}^m g^L_{n-k}] \\ f^R(T_{n-l};t,z)+ t(1-tz)[l-\sum_{k=1}^l g^R_{n-k}] \end{cases} 
\end{align*} 

suppose, $f^L(T_{n-m};t,z)=f^R(T_{n-l};t,z)$ then we obtain,
\begin{align} \label{eq:equation}
\sum_{k=1}^m g^L_{n-k}-\sum_{k=1}^l g^R_{n-k}=m-l
\end{align}

If $m=l=1$, i.e., both trees $T_{n-m}$ and $T_{n-l}$ are obtained from $T_n$ by a single grafting operation. But, since they have two different structures, the corresponding polynomials for the rooted trees $g^L_{n-1}$ and $g^R_{n-1}$ are distinct, because in \cite{tutte} it is shown that Tutte-like polynomial for rooted trees is graph invariant. Hence, $f^L(T_{n-1};t,z)$ and $f^R(T_{n-1};t,z)$ are distinct. So, the trees at the same level that are obtained from their parent through one grafting operation are distinct.

Finally, suppose we have $m\neq l$. Let's define $y_i=g_i-1$ for all polynomials $g_i(t,z)$. Now, using \eqref{eq:equation} we have,
\begin{align} \label{eq:equation2}
\sum_{k=1}^m y^L_{n-k}-\sum_{k=1}^l y^R_{n-k}=0
\end{align}

The highest degree term corresponds to the rooted trees, resulted by eliminating the edges $e$ and $e'$ and putting the common node between these two edges as a root. Also, the highest degree terms are resulted from the subtrees associated to $y^L_i$ and $y^R_i$ and no other proper subsets of these trees. Hence, from \eqref{eq:equation2} and assuming that original tree $T_n$ has the size $|E|$, then we can conclude, 
\begin{align} \label{eq:highestterm}
t^{|E|-2}\sum_{k=1}^m (1+z)^{L_{n-k}}=t^{|E|-2}\sum_{k=1}^l (1+z)^{R_{n-k}}
\end{align}

where $L_{n-k}$ and $R_{n-k}$ are non-negative integer powers, which show the largest number of internal edges for each tree associated to polynomials $y^L_{n-k}$ and $y^R_{n-k}$.

Equation \eqref{eq:highestterm} should hold for all values of $t$ and $z$. Let's set $t=1$ and $z=0$, we obtain $m=l$, a contradiction.

\section{Proof of Lemma~\ref{lem:specialbranches}} \label{app:specialbranches}

The second highest degree term consists of those subtrees of $T$ with exactly one edge missing. This edge should be a leaf edge, since otherwise the resulting structure does not become a tree. Now, consider those edges that are connected to the deleted leaf edge. They either become leaf or keep their previous states. Hence, the second highest degree term has the form $t^{|E|-1}(\alpha(1+z)^{|I|-1}+\beta(1+z)^{|I|})$, where $t^{|E|-1}$ is due to the deleted edge. Also, $\alpha$ is the number of resulted subtrees with $|I|-1$ internal edges, and $\beta$ is the number of those subtrees with $|I|$ internal edges. 
Hence, the number of leaf edges with no leaf siblings is $\alpha$.


\end{document}